\documentclass{LMCS}
\def\doi{7 (3:14) 2011}
\lmcsheading%
{\doi}
{1--24}
{}
{}
{Nov.~17, 2010}
{Sep.~21, 2011}
{}   
\usepackage{latexsym}
\usepackage{amsmath}
\usepackage{amssymb}
\usepackage{amsthm}
\usepackage{stmaryrd}
\usepackage{graphicx}
\usepackage{psfrag}
\usepackage{color}
\usepackage{enumerate}
\usepackage{ucs}
\usepackage{bbm}
\usepackage{hyperref}
\usepackage[utf8x]{inputenc}

\graphicspath{{figs/}}

\newcommand{\N}{\mathbbm{N}}

\newcommand{\betrag}[1]{\left\vert #1 \right\vert}
\newcommand{\st}{\mathbin |}

\newcommand{\GFtwo}{\mathbbm{F}_2}
\newcommand{\GFm}{\mathbbm{F}_m}
\newcommand{\GFp}{\mathbbm{F}_p}

\newcommand{\logn}{\mathop\mathrm{ln}}

\newcommand{\EF}{Eh\-ren\-feucht-Fraïssé}

\newcommand{\BPIFP}{\BP\IFP}
\newcommand{\Rescher}{\mathcal{J}}
\newcommand{\BPIFPJ}{\BPIFP(\Rescher)}
\newcommand{\BPFO}{\BP\FO}
\newcommand{\RFO}{\logic{RFO}}

\newcommand{\les}{\leqslant}
\renewcommand{\phi}{\varphi}
\newcommand{\bigmid}{\;\big|\;}

\newcommand{\logic}[1]{\textsf{\upshape #1}}
\newcommand{\LL}{\logic L}
\newcommand{\FO}{\logic{FO}}
\newcommand{\FOC}{\logic{FO+C}}
\newcommand{\SO}{\logic{SO}}
\newcommand{\MSO}{\logic{MSO}}
\newcommand{\TC}{\logic{TC}}
\newcommand{\DTC}{\logic{DTC}}
\newcommand{\TCC}{\logic{TC+C}}
\newcommand{\DTCC}{\logic{DTC+C}}
\newcommand{\LFP}{\logic{LFP}}

\newcommand{\IFP}{\logic{IFP}}
\newcommand{\IFPC}{\logic{IFP+C}}
\newcommand{\PFP}{\logic{PFP}}
\newcommand{\PFPC}{\logic{PFP+C}}
\newcommand{\Linf}{\LL^\omega_{\infty\omega}}
\newcommand{\Cinf}{\logic{C}^\omega_{\infty\omega}}
\newcommand{\Cinfs}{\logic{C}^s_{\infty\omega}}
\newcommand{\LK}{\logic K}
\newcommand{\PTIME}{\logic{P}}
\newcommand{\NP}{\logic{NP}}
\newcommand{\BPP}{\logic{BPP}}
\newcommand{\BP}{\logic{BP}}
\renewcommand{\P}{\logic{P}}

\newcommand{\CC}{\mathcal C}

\newcommand{\CO}{\mathcal O}
\newcommand{\CQ}{\mathcal Q}
\newcommand{\CR}{\mathcal R}
\newcommand{\CX}{\mathcal X}

\makeatletter
\renewcommand{\section}{\@startsection
  {section}%
  {1}%
  {0mm}%
  {-0.6\baselineskip}%
  {0.5\baselineskip}%
  {\normalfont\large\bfseries}%
}
\renewcommand{\subsection}{\@startsection
  {subsection}%
  {2}%
  {0mm}%
  {-0.6\baselineskip}%
  {0.5\baselineskip}%
  {\normalfont\bfseries}%
}
\makeatother

\newcounter{erom}
\renewcommand{\theerom}{(\roman{erom})}

\begin{document}

\title{Randomisation and Derandomisation in Descriptive Complexity Theory}
\author[K.~Eickmeyer]{Kord Eickmeyer}	%
\address{Humboldt-Universität zu Berlin, Institut für Informatik,
Logik in der Informatik\newline Unter den Linden 6, 
10099 Berlin, Germany}	%
\email{eickmeye@informatik.hu-berlin.de, grohe@informatik.hu-berlin.de}  %

\author{Martin Grohe}	%
%
%
%
%
%
%
%
%


\begin{abstract}
  We study probabilistic complexity classes and questions of
  derandomisation from a logical point of view. For each logic $\LL$
  we introduce a new logic $\BP\LL$, \emph{bounded error probabilistic
    $\LL$}, which is defined from $\LL$ in a similar way as the
  complexity class $\BPP$, bounded error probabilistic polynomial time,
  is defined from $\PTIME$.

  Our main focus lies on questions of derandomisation, and we prove
  that there is a query which is definable in $\BP\FO$, the
  probabilistic version of first-order logic, but not in $\Cinf$,
  finite variable infinitary logic with counting. This implies that
  many of the standard logics of finite model theory, like transitive
  closure logic and fixed-point logic, both with and without counting,
  cannot be derandomised. Similarly, we present a query on
  ordered structures which is definable in $\BP\FO$ but not in monadic
  second-order logic, and a query on additive structures which is
  definable in $\BP\FO$ but not in $\FO$. The latter of these queries
  shows that certain uniform variants of $\logic{AC}^0$ (bounded-depth
  polynomial sized circuits) cannot be derandomised. These results are
  in contrast to the general belief that most standard complexity
  classes can be derandomised.

  Finally, we note that $\BP\IFPC$, the probabilistic version of
  fixed-point logic with counting, captures the complexity class
  $\BPP$, even on unordered structures.
\end{abstract}

\maketitle

\section{Introduction}
The relation between different modes of computation --- deterministic,
nondeterministic, randomised --- is a central topic of computational
complexity theory. The $\PTIME$ vs. $\NP$ problem falls under this
topic, and so does a second very important problem, the relation
between randomised and deterministic polynomial time. In technical
terms, this is the question of whether $\PTIME=\BPP$, where $\BPP$ is
the class of all problems that can be solved by a randomised
polynomial time algorithm with two-sided errors and bounded error
probability.  This question differs from the question of whether $\PTIME=\NP$
in that most complexity theorists seem to believe that the classes
$\PTIME$ and $\BPP$ are indeed equal. This belief is supported by deep
results due to Nisan and Wigderson \cite{niswid94} and Impagliazzo and
Wigderson \cite{impwid97}, which link the derandomisation question to
the existence of one-way functions and to circuit lower
bounds; cf. also~\cite{impag06}. Similar derandomisation questions are studied for other
complexity classes such as logarithmic space, and it is believed that
derandomisation is possible for these classes as well.

Descriptive complexity theory gives logical descriptions of complexity
classes and thus enables us to translate complexity theoretic
questions into the realm of logic. While logical descriptions are
known for most natural deterministic and nondeterministic time and
space complexity classes, probabilistic classes such as $\BPP$ have
received very little attention in descriptive complexity theory
yet. In this paper, we study probabilistic complexity
classes and questions of derandomisation from a logical point of
view. For each logic $\LL$ we introduce a new logic $\BP\LL$,
\emph{bounded error probabilistic $\LL$}, which is defined from $\LL$
in a similar way as $\BPP$ is defined from $\PTIME$. The randomness is
introduced to the logic by letting formulas of vocabulary $\tau$ speak
about \emph{random expansions} of $\tau$-structures to a richer
vocabulary $\tau\cup\rho$. We also introduce variants $\logic{RL}$,
$\logic{co-RL}$ with one-sided bounded error and $\logic{PL}$ with
unbounded error, corresponding to other well known complexity classes.

Our main technical results are concerned with questions of
derandomisation. By this we mean upper bounds on the expressive power
of randomised logics in terms of classical logics. Trivially, $\BP\LL$
is at least as expressive as $\LL$, and if the two logics are equally
expressive, then we say that $\BP\LL$ \emph{derandomisable}. More generally,
if $\LL'$ is a (deterministic) logic that is at least as expressive as
$\BP\LL$, then we say that $\BP\LL$ derandomisable \emph{within} $\LL'$.
We prove that $\BP\FO$, bounded error probabilistic first-order logic, is not
derandomisable within $\Cinf$, finite
variable infinitary logic with counting.
This implies that many of the standard
logics of finite model theory, like transitive closure logic and
fixed-point logic, both with and without counting, cannot be
derandomised. Note that these results are in contrast to the
general belief that most standard complexity classes can be
derandomised.

We then investigate whether $\BPFO$ can be derandomised on classes of
structures with built-in relations, such as ordered structures and
arithmetic structures. We prove that $\BPFO$ cannot be derandomised
within $\MSO$, monadic second-order logic, on structures with built-in
order. Furthermore, $\BPFO$ cannot be derandomised on structures with
built-in order and addition. Interestingly and nontrivially, $\BPFO$
can be derandomised within $\MSO$ on structures with built-in order
and addition.  Behle and Lange \cite{bl06} showed that the expressive
power of $\FO$ on classes of ordered structures with certain
predefined relation symbols corresponds to uniform subclasses of
$\logic{AC}^0$, the class of problems decidable by circuit families of
bounded depth, unbounded fan-in and polynomial size. In fact, for any
set $\mathcal{R}$ of built-in relations they show that
$\FO[\mathcal{R}]$ captures $\FO[\mathcal{R}]$-uniform $\logic{AC}^0$.
Arguably the most intensively studied uniformity condition on
$\logic{AC}^0$ is \emph{dlogtime-uniform} $\logic{AC}^0$, which
corresponds to $\FO[+,\times]$, first-order logic with built-in
arithmetic (Barrington et al. \cite{bis90}). The question of whether
dlogtime-uniform $\logic{BPAC}^0$ can be derandomised is still open,
but there is a conditional derandomisation by Viola
\cite{vio04}. There are less uniform variants of $\logic{BPAC}^0$ that
can be proved to be derandomisable by standard arguments;
cf.~\cite{abenor84}. We prove that the more uniform $\FO[+]$-uniform
$\logic{AC}^0$ is not derandomisable. This raises the question of how
weak uniformity must be for derandomisation to be possible.

In the last section of this paper, we turn to more standard questions of
descriptive complexity theory. We prove that $\BP\IFPC$, the probabilistic
version of fixed-point logic with counting, captures the complexity class
$\BPP$, even on unordered structures. For ordered structures, this result is a
direct consequence of the Immerman-Vardi Theorem \cite{imm86,var82}, and for
arbitrary structures it follows from the observation that we can define a
random order with high probability in $\BP\IFPC$. Still, the result is
surprising at first sight because of its similarity with the open question of
whether there is a logic capturing $\PTIME$, and because it is believed that
$\PTIME=\BPP$. The caveat is that the logic $\BP\IFPC$ does not have an
effective syntax and thus is not a ``logic'' according to Gurevich's
\cite{gur88} definition underlying the question for a logic that captures
$\PTIME$. Nevertheless, we believe that $\BP\IFPC$ gives a completely adequate
description of the complexity class $\BPP$, because the definition of $\BPP$
is inherently ineffective as well (as opposed to the definition of $\PTIME$ in
terms of the decidable set of polynomially clocked Turing machines). We obtain
similar descriptions of other probabilistic complexity classes. For example,
randomised logspace is captured by the randomised version of deterministic
transitive closure logic with counting.

\subsection*{Related work}

As mentioned earlier, probabilistic complexity classes such as $\BPP$
have received very little attention in descriptive complexity
theory. There is an unpublished paper due to Kaye \cite{kay02} that
gives a logical characterisation of $\BPP$ on ordered
structures. M\"uller \cite{mueller08} and Montoya (unpublished) study
a logical $\BP$-operator in the context of parameterised complexity
theory. What comes closest to our work ``in spirit'' and also in some
technical aspects is Hella, Kolaitis, and Luosto's work on
\emph{almost everywhere equivalence} \cite{hkl96}, which may be viewed
as a logical account of average case complexity in a similar sense
that our work gives a logical account of randomised complexity. There
is a another logical approach to computational complexity, known as
implicit computational complexity, which is quite different from
descriptive complexity theory. Mitchell, Mitchell, and Scedrov
\cite{mitmitsce98} give a logical characterisation of $\BPP$ by a
higher-order typed programming language in this context.

Let us emphasise that the main purpose of this paper is not the
definition of new probabilistic logics, but an investigation of these
logics in a complexity theoretic context.

\section{Preliminaries}
\subsection{Structures and Queries}
\label{sec:sandq}

A \emph{vocabulary} is a finite set $\tau$ of
relation symbols of fixed arities. A $\tau$-\emph{structure} $A$
consists of a finite set $V(A)$, the \emph{universe} of the structure, and, for
all $R\in\tau$, a
relation $R(A)$ on $A$ whose arity matches that of
$R$. Thus we only consider \emph{finite} and \emph{relational}
structures. Let $\sigma,\tau$ be vocabularies with
$\sigma\subseteq\tau$. Then the \emph{$\sigma$-restriction} of a
$\tau$-structure $B$ is the $\sigma$-structure $B|_\sigma$ with
universe $V(B|_\sigma):=V(B)$ and relations $R(B|_\sigma):=R(B)$ for
all $R\in\sigma$. A \emph{$\tau$-expansion} of a $\sigma$-structure
$A$ is a $\tau$-structure $B$ such that $B|_\sigma=A$.
For every class $\CC$
of structures, $\CC[\tau]$ denotes the class
of all $\tau$-structures in $\CC$. A \emph{renaming} of a vocabulary $\tau$ is
a bijective mapping $r$ from $\tau$ to a vocabulary $\tau'$ such that for all
$R\in\tau$ the relation symbol $r(R)\in\tau'$ has the same arity as $R$. If
$r:\tau\to\tau'$ is a renaming and $A$ is a $\tau$-structure then $A^r$ is the
$\tau'$-structure with $V(A^r):=V(A)$ and $r(R)(A^r):=R(A)$ for all
$R\in\tau$.

We let $\les$, $+$ and $\times$ be distinguished relation symbols of
arity two, three and three, respectively. Whenever any of these relations symbols
appear in a vocabulary $\tau$, we demand that they be interpreted by a
linear order and ternary addition and multiplication relations,
respectively, in all $\tau$-structures. To be precise, let $[a,b]$ be
the set $\{a,a+1,\ldots,b\}$ for $a \leq b \in \N$, and denote by
$\mathcal{N}_n$ the $\{\les, +,\times\}$-structure with
\begin{align*}
V(\mathcal{N}_n) &= [0,n-1],
&
{\les}(\mathcal{N}_n) &= \{ (a,b) \st a \les b \} \text{ and}
\\
{+}(\mathcal{N}_n) &= \{ (a,b,c) \st a + b = c \},
&
{\times}(\mathcal{N}_n) &= \{ (a,b,c) \st a\cdot b = c \}.
\end{align*}
We demand $A|_{\{\les,+,\times\}\cap \tau} \cong
(\mathcal{N}_{\betrag{A}})|_{\{\les,+,\times\}\cap \tau}$ for all
$\tau$-structures $A$. We call structures whose vocabulary contains
any of these relation symbols \emph{ordered}, \emph{additive} and
\emph{multiplicative}, respectively. We say that a formula
$\varphi(x)$ with exactly one free variable $x$ \emph{defines an
  element} if in every structure it is satisfied by exactly one
element. Since we may identify the elements of an ordered structure
uniquely with natural numbers it makes sense to say, e.g., that
``$\varphi(x)$ defines a prime number'' or ``$\varphi(x)$ defines a
number $\leq \log^{O(1)} \betrag{A}$'', and we will sometimes do so.

On ordered structures, every fixed natural number $i$ can be defined
in first-order logic by a formula $\varphi_{i\text{-th}}$ using only
three variables as follows:
\[
\begin{split}
  \varphi_{0\text{-th}}(x) &:= \forall y\,x \leq y
  \\
  \varphi_{(n+1)\text{-th}}(x) &:=
  \exists y\forall z\,\big(\varphi_{n\text{-th}}(y)\wedge \neg(x \dot =
  y)
  \wedge y \leq x\wedge 
  \\
  &\qquad((y \leq z \wedge z \leq x)\rightarrow
  (y \dot = z \vee y \dot = z))\big).
\end{split}
\]
Because the ordering may be defined using the addition relation, the
same holds true on additive structures, again using only three
variables.

A \emph{$k$-ary $\tau$-global relation} is a mapping $\CR$ that
associates a $k$-ary relation $\CR(A)$ with each $\tau$-structure
$A$. A $0$-ary $\tau$-global relation is usually called a
\emph{Boolean} $\tau$-global relation. We identify the two $0$-ary
relations $\emptyset$ and $\{()\}$, where $()$ denotes the empty
tuple, with the truth values $\mathsf{false}$ and $\mathsf{true}$,
respectively, and we identify the Boolean $\tau$-global relation $\CR$
with the class of all $\tau$-structures $A$ with
$\CR(A)=\mathsf{true}$.  A \emph{$k$-ary $\tau$-query} is a $k$-ary
$\tau$-global relation $\CQ$ preserved under isomorphism, that is, if
$f$ is an isomorphism from a $\tau$-structure $A$ to a
$\tau$-structure $B$ then for all $\vec a\in V(A)^k$ it holds that
$\vec a\in\CQ(A)\iff f(\vec a)\in\CQ(B)$.

\subsection{Logics}
\label{sec:logic}
A \emph{logic} $\LL$ has a \emph{syntax} that assigns a set
$\LL[\tau]$ of \emph{$\LL$-formulas of vocabulary $\tau$} with each
vocabulary $\tau$ and a \emph{semantics} that associates a $\tau$-global
relation $\CQ^{\LL[\tau]}_\phi$ with
every formula $\phi\in\LL[\tau]$ such that for all vocabularies $\sigma,\tau,\tau'$ the following three conditions are
satisfied:
\begin{enumerate}[(1)]
\item For all $\phi\in\LL[\tau]$ the global relation $\CQ^{\LL[\tau]}_\phi$ is
  a $\tau$-query.
\item If $\sigma\subseteq\tau$ then $\LL[\sigma]\subseteq\LL[\tau]$, and for all
  formulas $\phi\in\LL[\sigma]$ and all $\tau$-structures $A$ it holds that
  $
  \CQ^{\LL[\sigma]}_\phi(A|_\sigma)=\CQ^{\LL[\tau]}_\phi(A).
  $
\item If $r:\tau\to\tau'$ is a renaming, then for every formula
  $\phi\in\LL[\tau]$ there is a formula $\phi^r\in\LL[\tau']$ such that for
  all $\tau$-structures $A$ it holds that
  $
  \CQ^{\LL[\tau]}_\phi(A)=\CQ^{\LL[\tau']}_{\phi^r}(A^r).
  $
\end{enumerate}
Condition (ii) justifies dropping the vocabulary $\tau$ in the
notation for the queries and just write $\CQ^\LL_\phi$.  For a
$\tau$-structure $A$ and a tuple $\vec a$ whose length matches the arity of
$\CQ_\phi^\LL$, we usually write $A\models_\LL\phi[\vec a]$ instead of $\vec
a\in\CQ^\LL_\phi(A)$. If $\CQ^\LL_\phi$ is a $k$-ary query, then we call $\phi$ a
\emph{$k$-ary} formula, and if $\CQ^\LL_\phi$ is Boolean, then we call
$\phi$ a \emph{sentence}. Instead of $A\models_\LL\phi[()]$ we just write
$A\models_\LL\phi$ and say that $A$ \emph{satisfies} $\phi$. We omit the index
$\LL$ if $\LL$ is clear from the context.

A query $\CQ$ is \emph{definable} in a logic $\LL$ if there is an
$\LL$-formula $\phi$ such that $\CQ=\CQ_\phi^\LL$. Two formulas
$\phi_1,\phi_2\in\LL[\tau]$ are \emph{equivalent} (we write
$\phi_1\equiv\phi_2$) if they define the same query. We say that a logic
$\LL_1$ is \emph{weaker} than a logic $\LL_2$ (we write
$\LL_1\leqq\LL_2$) if every query definable in $\LL_1$ is also
definable in $\LL_2$. Similarly, we define it for $\LL_1$ and $\LL_2$ to be
\emph{equivalent} (we write $\LL_1\equiv\LL_2$) and for $\LL_1$ to be
\emph{strictly weaker} than $\LL_2$ (we write $\LL_1\lneqq\LL_2$). The
logics $\LL_1$ and $\LL_2$ are \emph{incomparable} if neither
$\LL_1\leqq\LL_2$ nor $\LL_2\leqq\LL_1$.

\begin{rem}
  Our notion of logic is very minimalistic, usually logics are
  required to meet additional conditions (see \cite{ebb85} for a
  thorough discussion). In particular, we do not require the syntax of
  a logic to be effective. Indeed, the main logics studied in this
  paper have an undecidable syntax. Our definition is in the tradition
  of abstract model theory~(cf. \cite{barfef85}); proof theorists tend
  to have a different view on what constitutes a logic.
\end{rem}

We assume that the reader has heard of the standard
logics studied in finite model theory, specifically \emph{first-order logic}
$\FO$, \emph{second-order logic} $\SO$ and its fragments $\Sigma^1_k$,
\emph{monadic second-order logic} $\MSO$, \emph{transitive closure logic}
$\TC$ and its \emph{deterministic} variant $\DTC$, \emph{least},
\emph{inflationary}, and \emph{partial fixed-point
  logic} $\LFP$, $\IFP$, and $\PFP$, and \emph{finite variable
  infinitary logic} $\Linf$. For all these logics except $\LFP$ there are also
\emph{counting versions}, which we denote by $\FOC$, $\TCC$, $\ldots$,
$\PFPC$ and $\Cinf$, respectively.
Only familiarity with first-order logic is required
to follow most of the technical arguments in this paper. The other logics are
more or less treated as ``black boxes''. We will say a bit more about
some of them
when they occur later.  The following diagram shows how the logics
compare in expressive power:
\begin{equation}
  \label{eq:expressive-power}
  \begin{array}{ccccccccccc}
    \FO&\lneqq&\DTC&\lneqq&\TC&\lneqq&\LFP\;\equiv\;\IFP&\lneqq&\PFP&\lneqq&\Linf\\
    \lneqq&&\lneqq&&\lneqq&&\lneqq&&\lneqq&&\lneqq\\
    \FOC&\lneqq&\DTCC&\lneqq&\TCC&\lneqq&\IFPC&\lneqq&\PFPC&\lneqq&\Cinf.
  \end{array}
\end{equation}
Furthermore, $\MSO$ is strictly stronger than $\FO$ and
incomparable with all other logics displayed in \eqref{eq:expressive-power}.

\subsection{Complexity theory}

We assume that the reader is familiar with the basics of computational
complexity theory and in particular the standard complexity classes such as
$\PTIME$ and $\NP$. Let us briefly review the class $\BPP$, \emph{bounded
  error probabilistic polynomial time}, and other probabilistic complexity
classes: A language $L\subseteq\Sigma^*$ is in $\BPP$ if there is a polynomial
time algorithm $M$, expecting as input a string $x\in\Sigma^*$ and a string
$r\in\{0,1\}^*$ of ``random bits'', and a polynomial $p$ such that for every $x\in\Sigma^*$ the
following two conditions are satisfied:
\begin{enumerate}[(i)]
  \item If $x\in L$, then $\Pr_{r\in\{0,1\}^{p(|x|)}}\big(M\text{
      accepts }(x,r)\big)\ge\frac{2}{3}$.
  \item If $x\not\in L$, then $\Pr_{r\in\{0,1\}^{p(|x|)}}\big(M\text{
      accepts }(x,r)\big)\le\frac{1}{3}$.
\end{enumerate}
In both conditions, the probabilities range over strings
$r\in\{0,1\}^{p(|x|)}$ chosen uniformly at random. The choice of the
error bounds $1/3$ and $2/3$ in (i) and (ii) is somewhat arbitrary,
they can be replaced by any constants $\alpha,\beta$ with $0<\alpha<\beta<1$ without changing the
complexity class. (To reduce the error probability of an algorithm we
simply repeat it several times with independently chosen random bits
$r$.)

Hence $\BPP$ is the class of all problems that can be solved by a
randomised polynomial time algorithm with bounded error
probabilities. $\logic{RP}$ is the class of all problems that can be
solved by a randomised polynomial time algorithm with bounded
one-sided error on the positive side (the bound $1/3$ in (ii) is
replaced by $0$), and $\logic{co-RP}$ is the class of all problems
that can be solved by a randomised polynomial time algorithm with
bounded one-sided error on the negative side (the bound $2/3$ in (i)
is replaced by $1$). Finally, $\logic{PP}$ is the class we obtain if
we replace the lower bound $\ge 2/3$ in (i) by $> 1/2$ and the upper
bound $\le 1/3$ in (ii) by $\le 1/2$. Note that $\logic{PP}$ is not a
realistic model of ``efficient randomised computation'', because there
is no easy way of deciding whether an algorithm accepts or rejects its
input. Indeed, by Toda's Theorem \cite{tod91}, the class $\logic{P}^{\logic{PP}}$
contains the full polynomial hierarchy.  By the Sipser-G\'acs
Theorem~(see \cite{lau83}), $\BPP$ is contained in the second level of
the polynomial hierarchy. More precisely,
$\BPP\subseteq\Sigma^p_2\cap\Pi^p_2$.  It is an open question whether
$\BPP\subseteq\NP$. However, as pointed out in the introduction, there
are good reasons to believe that $\BPP=\PTIME$.
\subsection{Descriptive complexity}
It is common in descriptive complexity theory to view complexity
classes as classes of Boolean queries, rather than classes of formal
languages. This allows it to compare logics with complexity
classes. The translation between queries and languages is carried out
as follows:
Let $\tau$ be a vocabulary, and assume that $\mathord{\les}\not\in\tau$. With
each ordered $(\tau\cup\{\les\})$-structure $B$ we can associate a
binary string $s(B)\in\{0,1\}^*$ in a canonical way. Then with each class
$\CC\subseteq\CO[\tau\cup\{\les\}]$ of ordered $\tau$ structures we
associate the language
$L(\CC):=\{s(B)\mid B\in\CC\}\subseteq\{0,1\}^*$. For a Boolean
$\tau$-query $\CQ$, let $\CQ_\les:=\big\{B\in\CO[\tau\cup\mathord\les]\bigmid
B|_\tau\in\CQ\big\}$ be the class of all ordered
$(\tau\cup\{\les\})$-expansions of structures in $\CQ$. We say that
$\CQ$ is \emph{decidable} in a complexity class $\LK$ if the language
$L(\CQ_{\les})$ is contained in $\LK$. We say that a logic $\LL$
\emph{captures} $\LK$ if for all Boolean queries $\CQ$ it holds that
$\CQ$ is definable in $\LL$ if and only if $\CQ$ is decidable in
$\LK$. We say that $\LL$ is \emph{contained} in $\LK$ if all Boolean
queries definable in $\LL$ are decidable in $\LK$.

\begin{rem}
  Just like our notion of ``logic'', our notion of a logic
  ``capturing'' a complexity class is very minimalistic, but
  completely sufficient for our purposes.  For a deeper discussion of
  logics capturing complexity classes we refer the reader to one of
  the textbooks \cite{ebbflu95,gklmsvvw07,imm99,lib04}.
\end{rem}

\section{Randomised logics}

Throughout this section, let $\tau$ and $\rho$ be disjoint
vocabularies. Relations over $\rho$ will be ``random'', and we will reserve the letter $R$ for relation symbols
from $\rho$. We
are interested in \emph{random $(\tau\cup\rho)$-expansions} of $\tau$-structures.
For a $\tau$-structure $A$, by $\CX(A,\rho)$ we denote
the class of all $(\tau\cup\rho)$-expansions of $A$. We view
$\CX(A,\rho)$ as a probability space with the uniform
distribution. Note that we can ``construct'' a random
$X\in\CX(A,\rho)$ by deciding independently for all $k$-ary
$R\in\rho$ and all tuples $\vec a\in V(A)^k$ with probability $1/2$
whether $\vec a\in R(X)$. Hence if $\rho=\{R_1,\ldots,R_k\}$, where $R_i$ is
$r_i$-ary, then a random $X\in \CX(A,\rho)$ can be described by random
bitstring of length $\sum_{i=1}^kn^{r_i}$, where $n:=|V(A)|$.
We are mainly interested in the probabilities
\[
\Pr_{X\in\CX(A,\rho)}(X\models\phi)
\]
that a random $(\tau\cup\rho)$-expansion of a $\tau$-structure $A$
satisfies a sentence $\phi$ of vocabulary $\tau\cup\rho$ of some logic.

\begin{defi}
  Let $\LL$ be a logic and $0\le\alpha\le\beta\le1$.
  \begin{enumerate}[(1)]
  \item A formula
    $\phi\in\LL[\tau\cup\rho]$ that defines a $k$-ary query has an
    \emph{$(\alpha,\beta]$-gap} if for all $\tau$-structures $A$ and
    all $\vec a\in V(A)^k$ it holds that
    \[
    \Pr_{X\in\CX(A,\rho)}(X\models\phi[\vec a])\le\alpha\qquad\text{or}\qquad\Pr_{X\in\CX(A,\rho)}(X\models\phi[\vec
    a])>\beta.
    \]
    \item
      The logic $\P_{(\alpha,\beta]}\LL$ is defined as follows: For
      each vocabulary $\tau$,
      \[
      \P_{(\alpha,\beta]}\LL[\tau]:=\bigcup_{\rho}\big\{\phi\in\LL[\tau\cup\rho]\bigmid
      \phi\text{ has an $(\alpha,\beta]$-gap}\big\},
      \]
      where the union ranges over all vocabularies $\rho$ disjoint from
      $\tau$. To define the semantics, let
      $\phi\in\P_{(\alpha,\beta]}\LL[\tau]$. Let $k,\rho$ such that
      $\phi\in\LL[\tau\cup\rho]$ and $\phi$ is $k$-ary. Then for all $\tau$-structures $A$,
      \[
      \CQ^{\P_{(\alpha,\beta]}\LL}_\phi(A):=\big\{\vec a\in V(A)^k\bigmid
      \Pr_{X\in\CX(A,\rho)}(X\models_\LL\phi[\vec a])>\beta\big\}.
      \]
   \end{enumerate}
\end{defi}

\noindent It is easy to see that for every logic $\LL$ and all
$\alpha,\beta$ with $0\le\alpha\leq\beta\le1$ the logic
$\P_{(\alpha,\beta]}\LL$ satisfies conditions (i)--(iii) from
  Subsection~\ref{sec:logic} and hence is indeed a well-defined logic.
We let
\[
\P\LL:=\P_{(1/2,1/2]}\LL
\quad\text{and}\quad
\logic{RL}:=\P_{(0,2/3]}\LL
\quad\text{and}\quad
\BP\LL:=\P_{(1/3,2/3]}\LL.
\]
We can also define a logic $\P_{[\alpha,\beta)}\LL$ and let
$\logic{co-RL}:=\P_{[1/3,1)}\LL$. The following lemma, which is an
adaptation of classical probability amplification techniques to
randomised logics, shows that for reasonable $\LL$ the strength of the
logic $\P_{(\alpha,\beta]}\LL$ does not depend on the exact choice of
the parameters $\alpha,\beta$. This justifies the arbitrary choice of
the constants $1/3,2/3$ in the definitions of $\logic{RL}$ and
$\BP\LL$.

\begin{lem}
  Let $\LL$ be a logic that is closed under conjunctions and
  disjunctions. Then for all $\alpha,\beta$ with $0<\alpha<\beta<1$ it holds
  that
  $
  \P_{(0,\beta]}\LL\equiv\logic{RL}
  $
  and
  $
  \P_{(\alpha,\beta]}\LL\equiv\BP\LL.
  $
\end{lem}
\begin{proof}
Let $\tau$ an $\rho = \{R_1,\ldots,R_k\}$ be disjoint relational vocabularies and let
$\varphi \in \LL[\tau \cup \rho]$. For any $n \geq 1$ we define a new vocabulary
\[
\rho^{(n)} := \{ R^{(i)}_j \st 1 \leq i \leq n, 1 \leq j \leq k \},
\]
where the arity of $R^{(i)}_j$ is that of $R_j \in \rho$. Using
the renaming property with the renaming
\[
r^{(i)}:(\tau \cup \rho)  \to (\tau \cup \rho^{(n)})
\]
that leaves $\tau$ fixed and maps $R_j \in \rho$ to $R^{(i)}_j$ we get
sentences $\varphi^{(i)}$, which are the sentence $\varphi$ with every
occurrence of $R_j$ replaced by $R^{(i)}_j$. Since $\LL$ is closed
under conjunctions and disjunctions, for every $0 < l \leq n$ there is
an $\LL[\tau \cup \rho^{(n)}]$-sentence
\[
\varphi^{(n,l)} :=
\bigvee_{\substack{I \subseteq [n]\\\betrag{I} = l}} \bigwedge_{i\in I}
\varphi^{(i)}
\]
which is satisfied iff at least $l$ of the $\varphi^{(i)}$ are
satisfied. Notice that the $\varphi^{(i)}$ use distinct random
relations, so they are satisfied independently of each other.

Clearly, if $\Pr(X \models \varphi) = 0$ then also $\Pr(X \models
\varphi^{(n,l)}) = 0$, because we assumed $l \geq 1$. On the other
hand, if $\Pr(X \models \varphi)  > \beta$ for some $\beta \in (0,1)$,
then
\begin{align}
\Pr(X \models \varphi^{(n,1)}) &=
1 - (1-\Pr(X \models \varphi))^n
\\
&> 1 - (1-\beta)^n,
\end{align}
and this bound can be made arbitrarily close to $1$ by choosing $n$
sufficiently large. This proves the claim about $\logic{RL}$.

For $\logic{BPL}$, notice that if $\varphi$ has an
$(\alpha,\beta]$-gap for some any $0 < \alpha < \beta < 1$, then for
any $0 < \alpha' < \beta' < 1$ there is an $n \in \N$ such that
\[
\varphi^{(n,\lceil \frac{\beta-\alpha}{2}\rceil)}
\]
has an $(\alpha',\beta']$-gap. In fact, the Chernoff bound~(see, e.g.,
  \cite{motrag95}) gives very sharp estimates on $n$ in terms of
  $\alpha$, $\beta$, $\alpha'$ and $\beta'$, though we only need the
  mere existence of such an $n$ here.
\end{proof}

\subsection{First observations}

We start by observing that the syntax of $\BP\FO$ and thus of most other
logics $\BP\LL$ is undecidable. This follows easily from Trakhtenbrot's
Theorem (see \cite{ebbflu95} for similar undecidability proofs):

\begin{obs}
  For all $\alpha,\beta$ with $0\le\alpha<\beta<1$ and all
  vocabularies $\tau$ containing at least one at least binary relation
  symbol, the set $\BP_{(\alpha,\beta]}\FO[\tau]$ is undecidable.
\end{obs}
\begin{proof}[Proof Sketch]
  Assume for some $0 \le \alpha < \beta < 1$ and some $\tau$
  containing a binary relation symbol $E$ the set
  $\BP_{(\alpha,\beta]}\FO[\tau]$ is decidable.

  By Trakhtenbrot's Theorem (cf.~\cite[Thm. 7.2.1]{ebbflu95}), the
  satisfiability of a first-order formula $\psi \in \FO[\tau]$ on
  finite graphs is undecidable. Let $\mathcal{G}$ be the class of all
  graphs with exactly one isolated vertex, and let
  $\varphi_{\mathcal{G}}$ be a sentence defining $\mathcal{G}$
  on finite structures. By standard arguments, whether a formula is
  satisfiable in $\mathcal{G}$ or on is undecidable.

  Let $p = a\cdot 2^{-k} \in (\alpha,\beta)$ with $a \in \N$ be a dyadic
  rational in the interval $(\alpha,\beta)$, and let $R_1,\ldots,R_k$ be unary
  random relations. For every $S \subset [k]$, the sentence
  \[
  \psi_S := \exists x \left((\forall y\,\neg Exy) \wedge \bigwedge_{i \in S} R_ix \wedge
  \bigwedge_{i \not \in S} \neg R_ix\right)
  \]
  has satisfaction probability $2^{-k}$ in all structures in
  $\mathcal{G}$.
  Thus for a family
  $\mathcal{S}=\{S_1,\ldots,S_a\}$ of $a$ distinct subsets of $[k]$,
  the sentence
  \[
  \psi_{\mathcal{S}} := \bigvee_{S \in \mathcal{S}}\psi_S
  \]
  is satisfied with probability $p$ on such structures. But now the
  sentence
  \[
  \varphi_{\mathcal{G}} \rightarrow (\chi \wedge \psi_{\mathcal{S}})
  \]
  is in $\BP_{(\alpha,\beta]}\FO[\tau]$ if and only if $\chi$ is
  not satisfiable on $\mathcal{G}$.%
\end{proof}

For each $n$, let $S_n$ be the $\emptyset$-structure with universe
$V(S_n):=\{1,\ldots,n\}$. Recall the 0-1-law for first order logic
\cite{fag76,gleetal69}. In our terminology, it says that for each
vocabulary $\rho$ and each sentence $\phi\in\FO[\rho]$ it holds that
\[
\lim_{n\to\infty}\Pr_{X\in\CX(S_n,\rho)}(X\models\phi)\in\{0,1\}
\]
(in particular, this limit exists).  There is also an appropriate
asymptotic law for formulas with free variables. This implies that on
structures with empty vocabulary, $\logic{PFO}$ (and in particular
$\BPFO$) has the same expressive power as $\FO$. As there is also a
0-1-law for the logic $\Linf$ \cite{kolvar92b}, we actually get the
following stronger statement:

\begin{obs}
  \label{obs:emptyvocab}
  Every formula $\phi\in\logic{P}\Linf[\emptyset]$ is equivalent to a formula
  $\phi'\in\FO[\emptyset]$.
\end{obs}

\noindent
As $\FOC$ is strictly stronger than $\FO$ even on structures of empty
vocabulary, this observation implies that there are queries definable
in $\FOC$, but not in $\logic{(B)P}\Linf$.

Furthermore, the Sipser-Gács Theorem \cite{lau83} that
$\BPP\subseteq\Sigma^p_2\cap\Pi^p_2$, the fact that the fragment $\Sigma^1_2$
of second-order logic captures $\Sigma^p_2$ \cite{fag74,sto77}, and the observation that
$\BP\FO\leqq\BPP$ imply the following:

\begin{obs}\hspace{4cm}
\label{obs:sipsergacs}
  $
  \BP\FO\leqq\Sigma^1_2.
  $
\end{obs}
We will use Lautemann's proof of the Sipser-Gács Theorem in
section~\ref{sec:bpfovsmsoplus} in the context of monadic second-order
logic.

We close this section by observing that randomised logics
\emph{without} probability gaps are considerably more powerful than
their non-randomised counterparts:
\begin{obs}
Let $\mathcal{K}$ be a class of finite structures such that there is
a first-order formula $\varphi_c(x)$ defining a single element in each
structure of $\mathcal{K}$. Then every $\Sigma_1^1$-query on
$\mathcal{K}$ can be defined in $\logic{PFO}$.
\end{obs}
\begin{proof}
  Let $\varphi$ be a $\Sigma_1^1$-query on $\mathcal{K}$, i.e.,
  $\varphi$ is of the form $\exists X_1\cdots \exists X_k \psi$, where
  the $X_i$ are relation variables and $\psi$ is first-order. We
  replace each of the $X_i$ by a random relation $R_i$ of the same
  arity to get a new sentence $\varphi'$ and introduce an extra unary
  random relation $R_0$. Then $\varphi$ is equivalent to the
  $\logic{PFO}$-sentence
  \[
  \exists x(R_0 x \wedge \varphi_c(x)) \vee \varphi',
  \]
  because the first part is satisfied with probability exactly $1/2$.
\end{proof}

Toda's Theorem \cite{tod91} that the polynomial hierarchy is
contained in $\logic{P}^{\logic{PP}}$ suggests that, in fact, every
second-order query is definable in $\logic{PFO}$. However, Toda's
proof does not carry over easily to the $\logic{PFO}$-case.
Observation~\ref{obs:emptyvocab} suggests that some technical condition
such as definability of an element of the structure is necessary to
separate $\logic{PFO}$ from $\FO$ at all. One example of such a class
$\mathcal{K}$ is the class of all ordered structures, with
$\varphi_c(x)$ defining the minimum element.

\section{Separation results for $\BPFO$}
\label{sec:separation}

In this section we study the expressive power of the randomised logics
$\RFO$, $\logic{co-RFO}$, and $\BPFO$. Our main results are the
following:
\begin{iteMize}{$\bullet$}
\item $\RFO$ is not contained in $\Cinf$
\item $\BPFO$ is not contained in $\MSO$ on ordered structures
\item $\RFO$ is stronger than $\FO$ on additive structures
\end{iteMize}
A forteriori, the first and the third result also hold with $\BPFO$
instead of $\RFO$, and the constructions used in their proofs are also
definable in $\logic{co-RFO}$.

It turns out that we need three rather different queries to get these
separation results. For the first two queries this is obvious, because \emph{every} query on ordered structures is definable in
$\Cinf$. The third query (on additive structures) is readily seen to
be definable in $\MSO$. In fact, in
Section~\ref{sec:bpfovsmsoplus} we show the following:
\begin{iteMize}{$\bullet$}
\item Any $\BPFO$-definable query on additive structures can be
defined in $\MSO$.
\end{iteMize}

\subsection{$\RFO$ is not contained in $\Cinf$}
Formulas of the logic $\Cinf$ may contain arbitrary (not necessarily finite)
conjunctions and disjunctions, but only finitely many
variables, and counting quantifiers of the form
$
\exists^{\geq n}x\;\varphi
$
(``there exists at least $n$ $x$ such that $\phi$'').
For example, the class of finite
structures of even cardinality can be defined in this logic by the
sentence
\[
\bigvee_{k\geq 0}\left(\exists^{\geq 2k}x\,x\dot{=}x\right)\wedge \neg
\left(\exists^{\geq 2k+1}x\,x\dot{=}x\right).
\]

\begin{thm}\label{theo:bpfo-cinf}
  There is a class $\mathcal{TCFI}$ of structures that is definable in
  $\RFO$ and $\logic{co-RFO}$, but not in $\Cinf$.
\end{thm}

Recall that by Observation~\ref{obs:emptyvocab} there also is a
class of structures definable in $\FOC\le\Cinf$, but not in $\BPFO$.

Our proof of Theorem~\ref{theo:bpfo-cinf} is based on a well-known
construction due to Cai, Fürer, and Immerman \cite{cfi92}, who gave an
example of a Boolean query in $\PTIME$ that is not definable in
$\Cinf$. We modify their construction in a way reminiscent to a proof
by Dawar, Hella, and Kolaitis \cite{dhk95} for results on implicit
definability in first-order logic, and obtain a query $\mathcal{TCFI}$
definable in $\logic{(co-)RFO}$, but not in $\Cinf$. Just like in Cai,
Fürer and Immerman's original proof, the reason why $\Cinf$ can not
define our query $\mathcal{TCFI}$ is its inability to choose one
out of a pair of two elements. Using a random binary relation this can
-- with high probability -- be done in $\FO$.

We first review the construction of \cite{cfi92} and then show how to
modify it to suit our needs.  Given a graph $G = (V,E)$, Cai et
al. construct a new  graph $G'$, replacing all vertices and
edges of $G$ with certain gadgets. We shall call graphs $G'$
resulting in this fashion \emph{CFI-graphs}, and will from now on
restrict ourselves to connected 3-regular graphs $G$ and CFI-graphs resulting
from these.

\begin{figure}[htb]
  \begin{center}
    \input{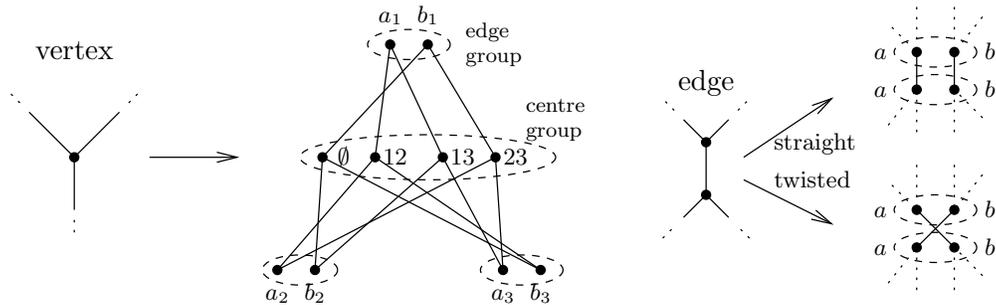}
  \end{center}
  \caption{The gadgets for CFI-graphs. Dashed ellipses indicate groups
    of equivalent vertices. Vertex labels are not part of the actual
    structure.}  \label{fig:cfigadget}
\end{figure}

The construction is as follows: For each vertex in $G$, we place a
copy of the gadget shown on the left of Figure~\ref{fig:cfigadget} in
$G'$. It has a group of four nodes (henceforth called \emph{centre
  nodes}) plus three pairs of nodes, which are to be thought of as ends
of the three edges incident with that node. For the time being, we
think of the pairs as ordered from $1$ to $3$ and distinguish between
the two nodes in each pair, say one of them is the $a$-node, the other
one being the $b$ node. Each of the four centre nodes is connected to
one node from each pair, and each of them to an even number of
$a$'s. To illustrate this, the centre nodes are labelled with the even
subsets of $\{1, 2, 3\}$. 
We also introduce an equivalence relation (or colouring, if you like)
of nodes as shown in Figure~\ref{fig:cfigadget}, so any isomorphism of
the gadget necessarily permutes nodes within each edge group and the
centre group.

\begin{figure}[htb]
  \begin{center}
    \input{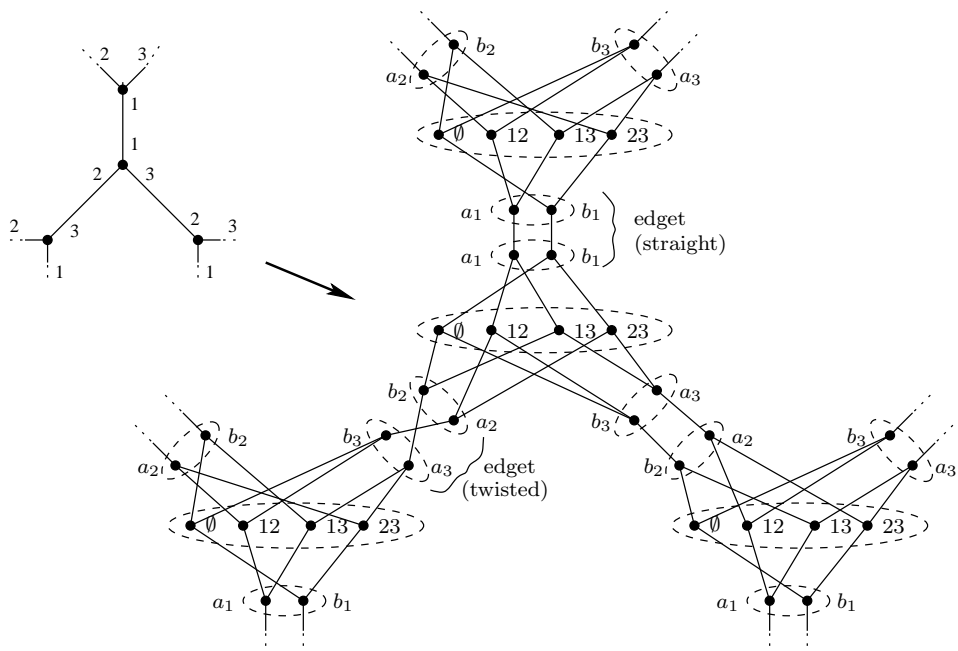}
  \end{center}
  \caption{The CFI-graph construction for a part of a graph. Edge and
    nodes labels are not part of the actual graph.}
  \label{fig:cfiexample}
\end{figure}

For each edge in $G$, we connect the $a$- and $b$-nodes in the
corresponding pairs as shown on the right of
Figure~\ref{fig:cfigadget}. We say an edge is ``twisted'' if the
$a$-node of one pair is connected to the $b$-node of the other and
vice versa. This completes our construction of $G'$. For definiteness,
when we speak of an \emph{edge group} we mean an equivalence class of
size two, and by a \emph{centre group} we mean one of size four. An
\emph{edget} is a pair of edge groups which form an edge gadget as on
the right of Figure~\ref{fig:cfigadget}. Figure~\ref{fig:cfiexample}
shows the result of applying this construction to a small subgraph (a
vertex with its three neighbours).

Without the $a$- and $b$-labels, we cannot decide which of the edges
have been twisted. In fact there are only two isomorphism classes of
CFI-graphs derived from $G$, namely those with an even number of edges
twisted and those with an odd number (we call the latter ones
\emph{twisted} CFI-graphs). This relies on the fact that isomorphisms
of the gadget on the left of Figure~\ref{fig:cfigadget} are exactly
those permutations swapping an even number of $a$'s and $b$'s. Since
we assume $G$ to be connected, we can twist edges along a path between
two nodes adjacent to twisted edges, reducing the number of twisted
edges by two; cf.~\cite[Lemma~6.2]{cfi92} for details.

By~\cite[Thm.~6.4]{cfi92}, if the original graph $G$ has no separator of
size at most $s$ then the two isomorphism classes of CFI graphs
derived from it can not be distinguished by a sentence $\varphi \in
\Cinfs$, i.e., by a $\Cinf$ sentence with at most $s$ distinct
variables. In \PTIME, on the other hand, twisted CFI-graphs can easily
be recognised: Choose exactly one node from each edge group and label
this one $a$ and the other one $b$. A centre node is connected to an
even number of $a$'s if and only if all four nodes in its centre group
are. In this case we call the centre group even, otherwise we call it
odd. Then a CFI-graph is twisted if and only if
\[
\label{eqn:twisted}
(\text{number of odd centre groups} +
 \text{number of twisted edgets})\text{ is odd}.
\]

We aim for a $\logic{(co-)RFO}$-sentence which defines exactly the
twisted connected 3-regular CFI-graphs. In view of the above
\PTIME-algorithm, we are done if we can
\begin{iteMize}{$\bullet$}
\item express connectedness of the graph,
\item count edgets and centre groups modulo two and
\item choose one representative from each centre group, edge group and edget.
\end{iteMize}

For counting modulo two and to get representatives for centre groups
and edgets, we augment the structures with a Boolean algebra in the
following way:
Let $\tau$ be the
vocabulary $\{E,\sim, <,\sqsubseteq,P,O\}$, with unary $P$ and $O$,
and binary $E$, $\sim$, $<$ and $\sqsubseteq$. Let $\mathcal{CFI}$ be
the class of structures $A$ such that
\begin{iteMize}{$\bullet$}
\item $E$ defines a 3-regular, connected CFI-graph on $V(A)
  \setminus P(A)$,
\item $(P(A), \sqsubseteq)$ is a Boolean algebra $\mathfrak{B}$, and
  $O$ is true exactly for its members of even cardinality
\item $<$ defines a linear order on the set of atoms of
  $\mathfrak{B}$ (and no other element of $A$ is $<$-related to any
  other).
\item $\sim$ defines an equivalence relation, where each equivalence
  class 
  \begin{iteMize}{$-$}
  \item contains one atom of $\mathfrak{B}$ and the nodes of one edget
  \item or contains one atom of $\mathfrak{B}$ and the nodes of one
    centre group
  \item or consists of a single non-atom of $\mathfrak{B}$.
  \end{iteMize}
  In particular, the number of atoms of the Boolean algebra
  $\mathfrak{B}$ is equal to the number of edgets plus the number of
  centre groups. Note also that we can distinguish the two edge groups
  in an edget because only nodes in the same edge group are connected
  to nodes in the same centre group.
\end{iteMize}  

\begin{thm}
  The class $\mathcal{CFI}$ is definable in $\FO$. The subclass
  $\mathcal{TCFI}$ of twisted CFI-graphs is definable in $\BPFO$
  but not in $\Cinf$.
\end{thm}
\begin{proof}
That $\mathcal{CFI}$ is definable is easy to establish, the only
subtlety  being that $\mathfrak{B}$ allows us to quantify over sets
of centre groups, which makes connectedness expressible.

The proof that $\mathcal{TCFI}$ is not definable in $\Cinf$ is the
same as in \cite{cfi92}; it is unaffected by the additional
structure. Note that because the atoms are ordered, the Boolean
algebra is rigid, i.e., it has no non-trivial automorphism, therefore
the isomorphism group of a CFI-graph is not changed by adding the
Boolean algebra.

It remains to show that twistedness can be defined in $\BPFO$. We
pick one vertex from each edge group by viewing a random binary
relation $R$ as assigning an $m$-bit number to each vertex, where $m$
is the number of atoms in the Boolean algebra. From each pair, we
choose the vertex with the smaller number, expressed by
\[
\xi(x) := %
\exists y\Big(%
x \sim y\wedge \exists
z\big(
\alpha(z) \wedge \neg Rxz \wedge Ryz \wedge \forall w
(w < z \rightarrow (Rxw \leftrightarrow Ryw))\big)\Big),
\]
where $\alpha(x)$ is an $\FO$-formula satisfied exactly by the
atoms of the Boolean algebra. It is easy to see that if the random
relation $R$ assigns a different set of atoms to the two vertices in
each edge group, then $\xi$ succeeds in picking
exactly one vertex from each edge group, and twistedness can then be checked
by looking at the $O$-predicate of the element of $\mathfrak{B}$ which
contains exactly the atoms equivalent to twisted centre groups or
twisted edgets.

To prove that the resulting formula has a large probability gap, we
need to establish a high probability of success only for structures in
the class $\mathcal{CFI}$, because this class is $\FO$-definable. But
in such structures, the probability that the two nodes of an edge
group are assigned the same number is $2^{-m}$, so by a union bound
the probability that we successfully pick one node from each group is
at least
\[
1-m2^{-m} \to 1
\]
because there are less than $m$ edgets. Furthermore, we can check in
$\FO$ whether there is an edge group whose members we can not
distinguish, and choose to invariably reject or accept in these cases,
resulting in an $\RFO$ or $\logic{co-RFO}$ sentence, respectively.
\end{proof}

\subsection{$\BPFO$ on ordered structures is not contained in $\MSO$}
\label{sec:bpfo}

In the presence of a linear order, \emph{any} query becomes
definable in $\Linf$, and the query $\mathcal{TCFI}$ becomes
definable even in $\FO$. However, randomisation adds expressive power to
\FO\ also on ordered structures:

\begin{thm}\label{theo:ord}
  There is a class
  $\mathcal B$ of ordered structures that is definable in $\BPFO$, but not in $\MSO$.
\end{thm}

Remember that monadic second-order logic MSO is the the fragment of
second-order logic that allows quantification over individual elements
and sets of elements.

\begin{figure}[tb]
  \label{fig:Kclass}
  \begin{center}
    \resizebox{0.7\textwidth}{!}{\input{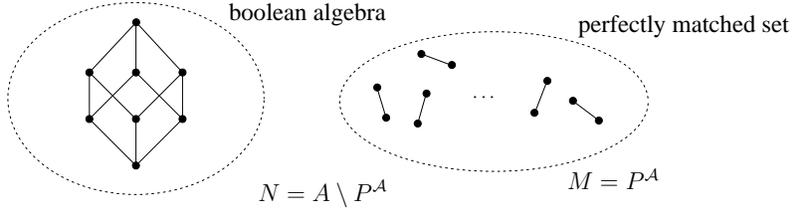}}
  \end{center}
  \caption{The structures in $\mathcal{B}$ contain a Boolean algebra
    and a perfectly matched set.}
\end{figure}
Let $\sigma_{EP\leq} := \{ \leq, E, P \}$, with binary relations
$\leq$ and $E$, and a unary predicate $P$. We define two classes
$\mathcal B'$, $\mathcal B$ of $\sigma_{EP\leq}$-structures
(cf. Figure~\ref{fig:Kclass}):

\noindent
$\mathcal B'$
  is the class of all $\sigma_{EP\leq}$-structures $A$ for
  which
  \begin{enumerate}[(1)]
   \item $E$ defines a perfect matching on the set $M := P(A)$
  \item the set $N := V(A) \setminus P(A)$ forms a Boolean
    algebra with the relation $E$ and
  \item no $x \in N$ and $y \in M$ are $E$-related
 \item $\leq$ defines a linear order on the whole structure, which puts
    the $M$ before the $N$ and orders $M$ in such a way that matched
    elements are always successive.
  \end{enumerate}
It is easy to see that the class $\mathcal B'$ is definable in \FO.
  $\mathcal B$ is the subclass of $\mathcal{B}'$ whose elements
  satisfy the additional condition
 \begin{equation}
     \label{eqn:sizecond}
    2^{\betrag{M}} \geq \betrag{N}^2.
 \end{equation}
 We will prove that $\mathcal B$ is definable in $\BPFO$, but not in $\MSO$.
 To prove that $\mathcal B$ is definable in $\BPFO$, we will use the following lemma:
\begin{lem}[Birthday Paradox]
\label{lem:birthday}
Let $m,n \geq 1$ and let $F : [n] \to [m]$ be a random function drawn
uniformly from the set of all such functions.
\begin{enumerate}[\em(1)]
\item For any $\epsilon_1 > 0$ and $c > 2 \logn \frac{1}{\epsilon_1}$ there is an $n_c \geq
1$ such that if $n > n_c$ and $m \leq \frac{n^2}{c}$ we have
\[
\Pr(F\text{ is injective}) \leq \epsilon_1
\]

\item For any $\epsilon_2 > 0$, if $m \geq \frac{n^2}{2\epsilon_2}$, then
\[
\Pr(F\text{ is injective}) \geq 1-\epsilon_2
\]
\end{enumerate}
\end{lem}
\begin{figure}[tb]
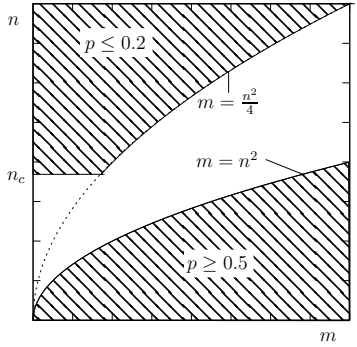

  \label{fig:birthday}
  \begin{center}
    \resizebox{0.3\textwidth}{!}{\input figs/birthday.pstex_t}
  \end{center}
  \caption{The Birthday Paradox with $\epsilon_1 = 0.2$, $\epsilon_2 =
    0.5$ and $c = 4$. Here, $p$ denotes $\Pr(f\text{ is injective})$.}
\end{figure}

\proof
For the first part, we note that
\[
\Pr(F\text{ injective})
  = \prod_{i = 0}^{n-1} \left(1-\frac i m\right)
\leq \prod_{i = 0}^{n-1} \exp\left(-\frac i m\right)
= \exp\left(-\frac{n(n-1)}{2m}\right).
\]
For the second part, note that
\[
\Pr(F\text{ not injective})
= \Pr\Big(F(i) = F(j) \text{ for all }i<j\Big)
\leq \sum_{i < j} \frac 1 m
= \binom{n}{2}\frac{1}{m}
\leq \frac{n^2}{2m}.\eqno{\qEd}
\]

\begin{proof}[Proof of Theorem~\ref{theo:ord}]
  To see that $\mathcal B$ is not definable in $\MSO$, we use two
  simple and well-known facts about $\MSO$. The first is that for
  every $q\ge0$ there are natural numbers $p,m$ such that for all
  $k\ge0$, a plain linear order of length $m$ is indistinguishable
  from the linear order of length $m+k\cdot p$ by $\MSO$-sentences of
  quantifier rank at most $q$. The same fact also holds for linear
  orders with a perfect matching on successive elements, because such
  a matching is definable in $\MSO$ anyway. The second fact we use is
  a version of the Feferman-Vaught Theorem
  (cf.~\cite[Thm.~1.5(ii)]{makowski04}):
  \begin{thm}
    Suppose two $\tau$-structures $U$ and $V$ satisfy the same
    $\MSO$-sentences of quantifier rank up to $q$, and let $W$ be
    another $\tau$-structure. Denote by $U\sqcup W$ (resp. $V\sqcup
    W$) the disjoint union of $U$ (resp. $V$) and $W$. Then
    $U\sqcup W$ and $V\sqcup W$ satisfy the same $\MSO$-sentences of
    quantifier rank up to $q$.
  \end{thm}
  The theorem also holds for the ordered disjoint union $\sqcup_{<}$
  instead of the disjoint union, but in our case the elements of the
  individual structures in the disjoint union are definable anyway.
  If we put these two facts together, we see that for every $q\ge 0$
  there are $p,m$ such that for all $k,n$ the structure
  $A\in\mathcal{B}$ with parts $M,N$ of sizes $m$, $n$, respectively,
  is indistinguishable from the structure $A'$ with parts of sizes
  $m+k\cdot p$ and $n$. We can easily choose $k$ and $n$ in such a way
  that $A\in\mathcal{B}$ and $A'\not\in\mathcal{B}$.

It remains to prove that $\mathcal{B}$ is definable in $\BPFO$. Consider the
sentence
\[
\varphi_{\text{inj}} := \forall x \forall y \Big(x\dot{=}y \vee Px \vee Py \vee \exists z \big(Pz \wedge \neg (Rxz \leftrightarrow Ryz)\big)\Big),
\]
which states that the random binary relation $R$, considered as a
function
\[
f: N \to \mathop\mathrm{Pow}(M),
\quad
x \mapsto \{ y \in M \st Rxy \}
\]
from $N$ to subsets of $M$, is injective. By the definition of $R$,
the function $f$ is drawn uniformly from the set of all such
functions. If we fix $\betrag{N}$, the probability for $f$ to be
injective increases monotonically with $\betrag{M}$. Furthermore,
for every structure in $\mathcal{B}'$, the size of $N$ and $M$ are a
power of two and an even number, respectively. Thus either
\[
2^{\betrag{M}} \leq \frac{1}{4}\betrag{N}^2
\quad
\text{or}
\quad
2^{\betrag{M}} \geq \betrag{N}^2,
\]
and this factor of $4$ translates into a probability gap for
$\varphi_{\text{inj}}$ in all sufficiently large structures in
$\mathcal{B}'$, by Lemma~\ref{lem:birthday} with $\epsilon_1 =
0.2$, $\epsilon_2 = 0.5$ and $c = 4$. The remaining
finitely many structures in $\mathcal{B}'$ can be dealt with
separately.
\end{proof}

\subsection{$\RFO$ is stronger than $\FO$ on additive structures}

Recall that an additive structure is one whose vocabulary
contains a ternary relation $+$, such that $A|_{+}$ is
isomorphic to $([0,\betrag{A}-1], \{(a,b,c) \st a+b=c\} )$.

\begin{thm}\label{theo:add}
  There is a class $\mathcal A$ of additive structures that is definable in
  $\RFO$ and $\logic{co-RFO}$, but not in $\FO$.
\end{thm}

Our proof uses the following result:
\begin{thm}[Lynch \cite{lyn82}]
For every $k \in \N$ there is an infinite set $A_k \subseteq \N$ and
a $d_k \in \N$
such that for all finite $Q_0,Q_1 \subseteq A_k$ with $\betrag{Q_0} =
\betrag{Q_1}$ or $\betrag{Q_0}, \betrag{Q_1} > d_k$ the structures
$(\N,{+},Q_0)$ and $(\N,{+},Q_1)$ satisfy exactly the same
$\FO$-sentences of quantifier rank at most $k$.
\end{thm}

Here $(\N,{+},Q_i)$ denotes a $\{{+},P\}$-structure with ternary ${+}$ and
unary $P$, where ${+}$ is interpreted as above and $P$ is interpreted by $Q_i$.
For a finite set $M \subseteq \N$ we denote by $\max M$ the maximum
element of $M$. By relativising quantifiers to the maximum element
satisfying $P$, we immediately get the following corollary:
\begin{cor}
\label{cor:lynchcor}
Let $k$, $A_k$, $d_k$, $Q_0$ and $Q_1$ be as above. Then the
(finite) structures
$([0,\max Q_0],{+},Q_0)$ and $([0,\max Q_1],{+},Q_1)$ satisfy exactly the same
$\FO$-sentences of quantifier rank at most $k$.
\end{cor}

We call a set $Q \subseteq \N$ \emph{sparse} if
$
\betrag{Q \cap \{n,\ldots,3n\}} \leq 1
$
for all $n \geq 0$. Note that if $Q$ is sparse and finite, then
$
\betrag{Q} \leq \log_3(\max Q) +1
$.
It is easy to see that there is an $\FO[\{{+},P\}]$-sentence $\varphi_\text{sparse}$ such that
\[
([0,\max Q],{+},Q) \models \varphi_\text{sparse}
\quad\Leftrightarrow\quad
Q\text{ is sparse}
\]
for all finite $Q \subseteq \N$.

\begin{proof}[Proof of Theorem~\ref{theo:add}]
  We define the following class of additive $\{+,P\}$-structures:
\[
\mathcal A=\{ ([0,\max Q],{+},Q)
\st
Q\text{ is finite, sparse and }\betrag{Q}\text{ is even}
\},
\]
with ${+}$ defined as usual. It follows immediately from
Corollary~\ref{cor:lynchcor} that $\mathcal A$ is not
definable in $\FO$.

It remains to prove that $\mathcal A$ is definable in $\logic{(co-)RFO}$. We
consider a binary random relation $R$ on $\mathcal{Q} = ([0,\max Q],
{+},Q)$ for some finite $Q \subseteq \N$.

Each element $a \in [0,\max Q]$ defines a subset of $Q$, namely the set
of $b \in Q$ for which
$(a,b) \in R(\mathcal{Q})$
holds. If $Q$ is a sparse set, it has
\[
2^{\betrag{Q}} \leq 2^{\log_3(\max Q) + 1} \leq \frac{\max
  Q}{2\mathop{\mathrm{ln}}(\max Q)}
\]
many subsets, and by standard estimates on the coupon collector's
problem (see, e.g., \cite{motrag95}; or use a union-bound argument),
if $\max Q$ is large enough, with high probability every subset of $Q$
is defined by some element of $[0,\max Q]$. We may check in $\FO$
whether this is actually the case. If so, we use the random relation
$R$ and the linear order induced by $+$ to check whether $Q$ is
even. Otherwise we reject (accept) to get an $\RFO$-
($\logic{co-RFO}$-)sentence.
\end{proof}

\section{$\BPFO$ is contained in $\MSO$ on additive structures}
\label{sec:bpfovsmsoplus}

In this section, we prove our first and only nontrivial
derandomisation result.  It complements the result of
Section~\ref{sec:bpfo} by saying that, on additive structures, every
$\BPFO$-sentence is equivalent to an $\MSO$-sentence.

\begin{thm}
  \label{thm:bpfoismsoonadd}
  Let $\tau$ be a finite relational vocabulary containing a ternary
  relation $+$ and let $\varphi$ be a $\BPFO[\tau]$-sentence. Then
  there exists an $\MSO$-sentence $\psi$ such that on additive
  structures $A$
  \[
  A \models \varphi \quad \Leftrightarrow \quad A \models \psi.
  \]
\end{thm}

We first use Nisan's pseudorandom generator for constant depth
circuits \cite{nis91} to reduce the number of random bits to
$\log^{O(1)}n$; throughout this section, $n$ will denote the size of
the input structure. We then derandomise the resulting formula
following Lautemann's argument in \cite{lau83}. The second-order
quantifier depth of the resulting $\MSO$ formula does not depend on
the input formula $\varphi$.

In $\MSO[+]$, one can define a multiplication relation
(see~\cite[Lemma~5.4]{Schweikardt06}) and thus
quantify over pairs of elements in $[0,\sqrt{n}]$.
We only need the existence of such a pairing function, a
slightly weaker form of which is made precise in the following lemma:
\begin{lem}[Pairing Lemma]
\label{lem:msopairing}
There are $\MSO[+]$-formulas $\varphi_p(x)$ and
$\varphi_{\langle\cdot,\cdot,\cdot\rangle}(x,y,z,w)$ such that on additive
structures $A$
\begin{iteMize}{$\bullet$}
\item $\varphi_p(x)$ defines a number $p$ satisfying
  \[
  \frac{\sqrt{\betrag{A}}}{2} \leq p \leq \sqrt{\betrag{A}}.
  \]
  Moreover, $p$ is a prime number.
\item For every $b, c < p$ there is a unique $m$ such that
  $\varphi_{\langle\cdot,\cdot,\cdot\rangle}(0,b,c,m)$ is
  satisfied. Furthermore, for every $m$ there is a unique tuple
  $(a,b,c) \in [0,p-1]^3$ such that
  $\varphi_{\langle\cdot,\cdot,\cdot\rangle}(a,b,c,m)$ is satisfied.
  Henceforth we write $m = \langle a,b,c \rangle$ for this.
\end{iteMize}
\end{lem}
\begin{proof}
In $\MSO[+]$, we may define a formulas $\varphi_{X = \langle x
  \rangle}(X,x)$ and $\varphi_{\mathrm{divides}}(x,y)$ stating that
$X$ is the set of multiples of $x$ and $x$ divides $y$,
respectively. We may thus check whether $x$ is a prime
number. Furthermore, we may define the set of powers of a prime
number $x$: It is the largest set containing only numbers whose only
prime divisor is $x$.

Then $p$ is the largest prime number whose set of powers contains at
least one element other that $0$ and itself. Any number $m \in [0,p^2-1]$
may be written as $m = bp + c$ with $b,c \in [0,p-1]$. Both $b$ and $c$
are definable in $\MSO[+]$; notice that $b$ is the largest divisor of
$m - c$ smaller than $p$, or $0$ if $m < p$. For $m \geq p^2$ we
define $m = \langle a, b, c \rangle$ with $a \in \{1,2,3\}$ and
$m-ap^2 = \langle 0,b,c\rangle$.
\end{proof}

Whenever we write $p$ in this section, we mean the $p$ defined by the
$\varphi_p$ above. The Pairing Lemma allows us to quantify over
binary relations on $[0,p-1] \cong \GFp$. In particular, we may define
addition and multiplication modulo $p$, i.e., there are
$\MSO[+]$-formulas $\varphi_+(x,y,z)$ and $\varphi_\times(x,y,z)$ such
that for $a,b,c \in \GFp$,
\[
A \models \varphi_+(a,b,c) \quad\Leftrightarrow\quad a + b \equiv c \pmod{p}
\]
and
\[
A \models \varphi_\times(a,b,c) \quad\Leftrightarrow\quad a\cdot b
\equiv c \pmod{p}.
\]

For the proof of Theorem~\ref{thm:bpfoismsoonadd} we may assume that
the $\BPFO$-sentence $\varphi$ contains only one random relation, say
$R$ of arity $r$.  In fact, using the formulas $\varphi_{i\text{-th}}$
defining the $i$-th element of an additive structure
(cf. section~\ref{sec:sandq}) we may pack several random relations
$R_1,\ldots,R_k$ of arities $r_1,\ldots,r_k$ into one random relation
$R$ of arity $r = 1+\max\{r_1,\ldots,r_k\}$ by replacing every
occurrence of $R_ix_1\ldots x_{r_i}$ by
\[
\exists y\,(\varphi_{i\text{-th}}(y)\wedge R\underbrace{y\ldots y}_{(r-r_i)\text{
    times}}x_1\ldots x_{r_i}).
\]

We first apply a result by Nisan \cite{nis91} to
reduce the number of random bits:
\begin{lem}
\label{lem:nisanmso}
For every $r, d \in \N$ and $\epsilon > 0$ there are $n_0 \in \N$ and
$\MSO[+]$-formulas $\varphi_l(x)$ and
$\varphi_{\mathrm{prg}}(S,x_1,\ldots,x_r)$, where $S$ is a set
variable, such that in every additive structure $A$ of size $n > n_0$,
\begin{iteMize}{$\bullet$}
\item $\varphi_l$ defines a number $l \leq \log^{O(1)} n$ and

\item if $\varphi$ is an $\FO[\tau \cup \{R\}]$-sentence of quantifier
  rank $\leq d$, where $\tau$ is some finite relational vocabulary and $R$
  is of arity $r$, then
  \[
  \betrag{
    \Pr_{X \in \CX(A,\{R\})}(X \models \varphi)
    -
    \Pr_{S \subseteq [l]}( A \models \varphi'(S) )
  } < \epsilon,
  \]
  where $\varphi'$ is the $\MSO[+]$-formula obtained from $\varphi$ by
  replacing every occurrence of $R\vec x$ by
  $\varphi_{\mathrm{prg}}(S,\vec x)$.
\end{iteMize}
\end{lem}
\begin{proof}
  For any fixed structure $A$ of size $n$ we may construct a
  polynomial-sized circuit $C_{\varphi,A}$ of depth $\leq d$ which
  describes the behaviour of $\varphi$ on $(\tau \cup
  \{R\})$-extensions of $A$. The circuit has $n^r$ inputs indexed by
  the elements of $V(A)^r$, and an input vector
  $\vec x$ denotes the $(\tau \cup \{R\})$-extension
  $B_{\vec x}$ of $A$ given by
  \[
  \vec a \in R(B_{\vec x}) \quad\text{iff}\quad
  x_{\vec a} = 1.
  \]
  Then $C_{\varphi,A}(\vec x)$ evaluates to $1$ iff $B_{\vec x}
  \models \varphi$.

  Nisan \cite{nis91} gave a pseudorandom generator for such circuits
  which hinges on the following lemma:
  \begin{lem}[restated from~{\cite[Lemma~2.2]{nis91}}]
    \label{lem:nisanlem}
    Let $\{C_n\}$ be a family of circuits of depth $d$ and polynomial
    size, let $m = m(n) = (\log n)^{d+3}$, $l = l(n)$ and suppose for
    each $n$ the sets $A^{(n)}_1,\ldots,A^{(n)}_n \subseteq
    [l]$ satisfy
    \begin{iteMize}{$\bullet$}
    \item $\betrag{A^{(n)}_i} = m$ for all $1 \leq i \leq n$ and
    \item $\betrag{A^{(n)}_i \cap A^{(n)}_j} \leq \log n$ for all $1
      \leq i \not = j \leq n$.
    \end{iteMize}
    Then
    \[
    \betrag{ \Pr(C_n(\vec x) = 0) - \Pr(C_n(\oplus_{i \in
        A_1}y_i,\ldots,\oplus_{i \in A_n}y_i) = 0)} \leq
    \frac{1}{n^c}
    \]
    for any $c \in \N$ and large enough $n$. Here, the first
    probability is taken uniformly over all strings $\vec x \in
    \{0,1\}^n$, whereas the second is taken uniformly over all strings
    $\vec y \in \{0,1\}^l$.
  \end{lem}
  The resulting pseudorandom generator is depicted in
  Figure~\ref{fig:nisanprg}. Families of sets $A^{(n)}_i$ satisfying
  the above conditions are called \emph{partial-$(\log
    n,m)$-designs}. Nisan gives a construction with $l = m^2 =
  \log^{O(1)} n$, which drastically reduces the size of the
  probability space, i.e., the number of random bits needed. We now
  show how his construction can be defined in $\MSO[+]$.

  \begin{figure}
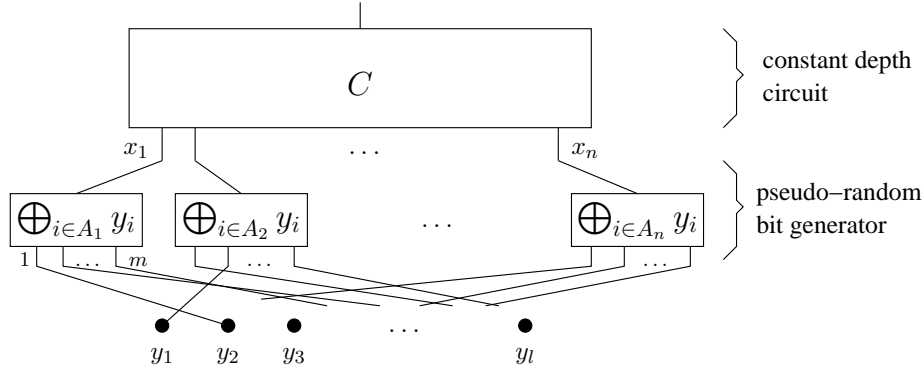

    \label{fig:nisanprg}
    \begin{center}
      \resizebox{0.8\textwidth}{!}{\input figs/nisanprg.pstex_t}
    \end{center}
    \caption{Nisan's pseudo-random bit generator. The sets $A_i \subseteq
      \{1,\ldots,l\}$ form a partial-$(\log n,m)$-design, i.e., they
      satisfy $\betrag{A_i} = m$ and $\betrag{A_i \cap A_j} \leq \log
      n$ for all $1 \leq i\not = j \leq n$.}
  \end{figure}

  On $[0,p-1]$, we may define a formula $\varphi_{\log}(x,y)$
  which is satisfied iff $x = \lceil \log_2 y \rceil$.  Using this and
  the fact that 
  \[
  2\lceil \log p \rceil - 1 \leq
  \lceil \log n \rceil \leq
  2\lceil \log p \rceil + 2,
  \]
  we let $\varphi_m(x)$ and $\varphi_l(x)$ be two formulas defining
  natural numbers $m$ and $l$ such that
  \begin{iteMize}{$\bullet$}
  \item $m$ is a prime number between $(r^2 \lceil \log n
    \rceil)^{d+3}$ and $2(r^2 (\lceil \log n \rceil + 3)^{d+3}$
  \item $l = m^2$
  \end{iteMize}

  Using the pairing function $\varphi_{\langle \cdot,\cdot,\cdot
    \rangle}$ we may assume that $R$ is a $3r$-ary relation which we
  only need to define for elements in $\GFp$. That is, we define
  $\varphi_{\mathrm{prg}}(S,x_1,\ldots,x_r)$ by
  \[
  \exists z_1 \cdots \exists z_{3r}
  \,
  x_1 = \langle z_1,z_2,z_3 \rangle \wedge
  \ldots
  \wedge x_r = \langle z_{3r-2},z_{3r-1},z_{3r} \rangle
  \wedge {\varphi'\negthinspace}_{\mathrm{prg}}(S,z_1,\ldots,z_{3r})
  \]
  The formula ${\varphi'\negthinspace}_{\mathrm{prg}}(S,\vec z)$ takes the parity of a
  subset of $S$ indexed by $\vec z$:
  \[
  {\varphi'\negthinspace}_{\mathrm{prg}}(S,\vec z) := \text{``}\betrag{S \cap \psi(A;\vec
    z)}\text{ is even}\text{''},
  \]
  where $\psi(x,\vec z)$ is an $\MSO[+]$-formula and $\psi(A;\vec z)
  := \{ x \st A \models \psi(x,\vec z) \}$; evenness may be expressed in $\MSO$ on
  ordered structures. By Lemma~\ref{lem:nisanlem}, we are done if we
  can define a formula $\psi(x,\vec z)$ such that
  \begin{enumerate}[(i)]
  \item $\psi(A;\vec z) \subseteq [l]$ for all $\vec z \in \GFp^{3r}$,
  \item $\betrag{\psi(A;\vec z)} = m$ for all $\vec z \in \GFp^{3r}$,
    and
  \item $\betrag{\psi(A;\vec z_1) \cap \psi(A;\vec z_2)} \leq \log n$ for
    all $\vec z_1 \not = \vec z_2 \in \GFp^{3r}$,
  \end{enumerate}
  which means the sets $\psi(A;\vec z)$ form a partial-$(\log n,
  m)$-design. We use the same construction as Nisan: We interpret the
  tuple $\vec z$ as a polynomial $f_{\vec z} \in \GFm[\xi]$ of degree
  $\leq \log n$. The set $\psi(A;\vec z)$ is then the graph of this
  polynomial, namely
  \[
  \psi(A;\vec z) = \{ (\xi,f_{\vec z}(\xi)) \st \xi \in \GFm \} \subseteq \GFm^2,
  \]
  and we identify $\GFm^2$ with $[l]$.
  We first encode the coefficients of $f_{\vec z}$
  into a set variable $X$ as follows: Consider the binary
  representations
  \[
  z_i = \sum_{j \geq 0} z_{i,j}2^j
  \quad\text{with }z_{i,j} \in \{0,1\}
  \]
  of the $z_i$. We can define an $\MSO[+]$-sentence
  $\varphi_{\mathrm{pack}}(\vec z,X)$ which holds iff $X$, interpreted
  as a binary relation over $\GFp$, holds exactly for pairs $(a,b)$
  with
  \[
  0 \leq a \leq \lceil \log p\rceil
  \quad\text{and}\quad
  b = \sum_{1 \leq i \leq 3r} z_{i,a}2^{i-1}.
  \]
  Thus for each $0 \leq a \leq \lceil \log p \rceil$ there is exactly one
  $b = b(a)$ with $(a,b) \in X$, and all $b$s are between $0$ and
  $2^{3r}$, and thus in $\GFm$ if $n$ is large enough. We may now
  define an $\MSO[+]$-sentence $\varphi_{\mathrm{eval}}(X,u,v)$ which,
  for these $X$s, holds iff
  \[
  v = f_{\vec y}(u) = \sum_{0 \leq a < \lceil \log p \rceil} b(a)u^a,
  \]
  with addition and multiplication according to $\GFm$. Putting these
  ingredients together, we define
  \[
  \psi(x,\vec z) = \exists X \exists u \exists v \,
  \text{``}0 \leq u,v < m\text{''}
  \wedge
  \varphi_{\mathrm{pack}}(\vec z, X)
  \wedge
  \varphi_{\mathrm{eval}}(X,u,v)
  \wedge
  \text{``}x = u\cdot m + v\text{''},
  \]
  which is easily verified to satisfy conditions (i) to (iii) above.
\end{proof}

So far we have reduced the number of random bits from $n^r$ to $l =
\log^{O(1)} n$, and these are conveniently packed into the first
$l$ bits of a single set variable $S$. We may now follow Lautemann's
proof \cite{lau83} to derandomise this sentence.
\begin{proof}[Proof of Theorem~\ref{thm:bpfoismsoonadd}]
  After applying Lemma~\ref{lem:nisanmso} we are left with
  $\MSO[+]$-sentences $\varphi_l$ and $\varphi'$ such that $\varphi_l$
  defines a number $l \leq \log^{O(1)} n$ and $\varphi'$ has a free
  set variable $S$.  We may assume that for all additive structures
  $A$,
  \begin{equation}
    \label{eqn:onelgap}
    \text{either }\Pr_{S \subseteq [l]}(A \models \varphi'(S)) < \frac{1}{l}
    \qquad
    \text{or }\Pr_{S \subseteq [l]}(A \models \varphi'(S)) > 1-\frac{1}{l},
  \end{equation}
  because otherwise we may use independent repetition and majority
  vote to obtain these bounds. To be precise, let $\chi(S,i,j)$ be
  defined by
  \[
  \chi(S,i,j) := (0 \leq i < l)\wedge (0 \leq j < l)\wedge \exists z
  (z \dot = i\cdot l + j \wedge Sz).
  \]
  That is, we divide the first $l^2$ bits of $S$ into $l$ blocks of
  $l$ bits each, and let $\chi(S,i,j)$ select the $i$-th bit of the
  $j$-th block. We replace each occurrence of $Sx$ in $\varphi'$ by
  $\chi(S,i,x)$ to obtain a formula ${\tilde\varphi}'(S,i)$. 
  Because $l$ is of order $\log^{O(1)} n$, we may quantify over pairs
  of elements of $[0,l-1]$, which allows us to express the formula
  \[
  \begin{split}
  {\bar\varphi}'(S) &= \text{``}{\tilde\varphi}'(S,i)\text{ holds for at
    least half of the }i \in [0,l-1]\text{''}
  \end{split}
  \]
  in $\MSO[+]$, e.g., by stating that there exists a matching $M$
  on $[0,l-1]$ such that
  \begin{iteMize}{$\bullet$}
  \item if $\{i,j\} \in M$, then exactly one of
    ${\tilde\varphi}'(S,i)$ and ${\tilde\varphi}'(S,j)$ holds and
  \item all $i \in [0,l-1]$ for which ${\tilde\varphi}'(S,i)$ does not
    hold are matched by $M$.
  \end{iteMize}
  Then ${\bar\varphi}'$ uses $l^2 = \log^{O(1)}n$ many bits of $S$, and by the
  Chernoff bound on the tails of the binomial distribution
  it satisfies~\eqref{eqn:onelgap}, even with $l$
  replaced by $l^2$ (details can be found in \cite[sec.~7.4]{ab09}).

  We identify subsets of $[l]$ with vectors in
  $\GFtwo^l$. Let $M \subseteq \GFtwo^l$ be the set of vectors for
  which $A \models \varphi'(S)$ holds. Equation~\eqref{eqn:onelgap}
  translates into
\[
\betrag{M} < \frac{\betrag{\GFtwo^l}}{l}
\quad\text{or}\quad
\betrag{M} > \left(1-\frac 1 l\right)\betrag{\GFtwo^l}.
\]
For a vector $\vec y \in \GFtwo^l$ we define
\[
\vec y \oplus M := \{ \vec x \oplus \vec y \st \vec x \in M \}
\]
to be the set $M$ translated by $\vec y$. We claim the following:
\begin{enumerate}[(a)]
\item If $\betrag{M} < \betrag{\GFtwo^l}/l$, then for every choice of
  vectors $\vec y_1,\ldots,\vec y_l$ we have
  \[
  \bigcup_{1 \leq i \leq l} (\vec y_i \oplus M) \not = \GFtwo^l.
  \]
\item If $\betrag{M} > (1-1/l)\betrag{\GFtwo^l}$, then there are
  vectors $\vec y_1,\ldots,\vec y_l$ such that
  \[
  \bigcup_{1 \leq i \leq l} (\vec y_i \oplus M) = \GFtwo^l.
  \]
\end{enumerate}
The first claim follows immediately from $\betrag{\vec y \oplus M} =
\betrag{M}$. For (b), assume that we randomly choose the
vectors $\vec y_i$ independently and uniformly from $\GFtwo^l$. For
any vector $\vec x \in \GFtwo^l$ we have
\[
\begin{split}
\Pr\left(\vec x \not\in  \bigcup (\vec y_i \oplus M)\right) 
&= \prod_i \Pr(\vec x \not\in \vec y_i \oplus M)
\\
&\leq \left(\frac 1 l \right)^l,
\end{split}
\]
by the independence of the $\vec y_i$. But then the expected number of
vectors \emph{not} in $\bigcup (\vec y_i \oplus M)$ is
\[
\begin{split}
\mathbbm{E}\left[ \betrag{\GFtwo^l \setminus \bigcup (\vec y_i
    \oplus M)} \right]
&= \sum_{\vec x \in \GFtwo^l} \Pr\left(\vec x \not\in  \bigcup
(\vec y_i \oplus M) \right) 
\\
&\leq \frac{\betrag{\GFtwo^l}}{l^l} = \left(\frac{2}{l}\right)^l < 1,
\end{split}
\]
so there must be a choice of $\vec y_i$s such that this number is
zero, i.e., $\bigcup (\vec y_i \oplus M) = \GFtwo^l$.

Again using the formula $\chi(S,i,j)$, we can pack the vectors $\vec
y_1,\ldots,\vec y_l$ into a single existentially quantified set
variable and check that $\bigcup (\vec y_i \oplus M) = \GFtwo^l$
as follows:
\[
\varphi'' = \exists Y \forall X \exists i\,
\varphi'(X \oplus \chi(Y,i,\cdot)),
\]
where $\varphi'(X \oplus \chi(Y,i,\cdot))$ is the formula $\varphi'(S)$
with every occurrence of $Sx$ replaced by
\[
(Xx \wedge \chi(Y,i,x)) \vee
(\neg Xx \wedge \neg\chi(Y,i,x)).
\]
Claims (a) and (b) imply that
\[
A \models \varphi''
\quad\Leftrightarrow\quad
\Pr(A \models \varphi'(S)) > 1-\frac 1 l,
\]
which completes the proof.
\end{proof}

\section{A logic capturing \BPP}

In this section, we prove that the logic $\BP\IFPC$ captures the
complexity class $\BPP$. Technically, the results of this section are
closely related to results in \cite{hkl96}.

Counting logics like $\FOC$ and $\IFPC$ are usually defined via
two-sorted structures, which are equipped with an initial segment of
the natural numbers of appropriate length. The expressive power of the
resulting logic turns out to be rather robust under changes in the
exact definition, see \cite{otto96} for a detailed survey of
this.
However, we will only need the limited counting ability provided by the
\emph{Rescher quantifier},
which goes back to a unary majority quantifier defined in
\cite{res62}, see \cite{otto96}.

We let $\FO(\Rescher)$ be the logic obtained from first-order logic by
adjoining a generalised quantifier $\Rescher$, the \emph{Rescher
  quantifier}. For any two formulas
$\varphi_1(\vec x)$ and $\varphi_2(\vec x)$, where $\vec x$ is a
$k$-tuple of variables, we form a
new formula
\[
\Rescher \vec x.\varphi_1(\vec x)\varphi_2(\vec x).
\]
Its semantics is defined by
\begin{multline}
A \models \Rescher\vec x.\varphi_1(\vec x)\varphi_2(\vec x)
\quad\text{ iff}
\\
\betrag{\{\vec{a}\in V(A)^k \st A \models
  \varphi_1[\vec{a}] \}}
\leq
\betrag{\{\vec{a}\in V(A)^k \st A \models
  \varphi_2[\vec{a}] \}}.
\end{multline}
The logic $\IFP(\Rescher)$ is defined similarly.

\begin{lem}\label{lem:ordlem}
  Let $R$ be a $6$-ary relation symbol.
  There is a formula $\phi_\le(x,y)\in\FO(\Rescher)[\{R\}]$ such that
  \[
  \lim_{n\to\infty}\Pr_{A\in X(S_n,\{R\})}\Big(\{(a,b)\bigmid
  A\models\phi_\le[a,b]\big\}\text{ is a linear order of }V(A)\Big)=1.
  \]
  (Recall that $S_n$ is the $\emptyset$-structure with universe
  $\{1,\ldots,n\}$. Thus $X(S_n,\{R\})$ just denotes the set of all
  $\{R\}$-structures with universe $\{1,\ldots,n\}$.)
\end{lem}

\begin{proof}
  We let
\[
\varphi_\leq(x,y) := \Rescher x_1\ldots x_5.Rx x_1\ldots x_5
\,
Ry x_1\ldots x_5.
\]
To see that $\phi_\le(x,y)$ defines an order with high probability, let $A$ be a
structure with universe $V(A)=\{1,\ldots,n\}$. For each $a\in V(A)$, let
\[
X_a := \betrag{ \{ \vec{a} \in V(A)^5 \st A \models
R a \vec{a}.
\} }
\]
Then $A \models \varphi_\leq(a,b)$ iff $X_a \leq X_b$,
and $\varphi_\leq$ linearly orders $A$ iff the $X_a$ are pairwise
distinct. But for $a \not = b \in V(A)$, the random variables $X_a$ and
$X_b$ are independent and each is binomially distributed with
parameters $p = 1/2$ and $m = n^5$, and thus
\[\begin{split}
\Pr(X_a = X_b) &= \sum_{k=0}^m \left(\frac{1}{2^m}\binom{m}{k}\right)^2
=\frac{1}{2^{2m}}\sum \binom{m}{k}^2
\\
&=\frac{1}{2^{2m}}\sum \binom{m}{k}\binom{m}{m-k}
=\frac{1}{2^{2m}}\binom{2m}{m} = \Theta\left(\frac{1}{\sqrt m}\right),
\end{split}\]
where the final approximation can be found, for example, in \cite{fellerI}.
The second part now follows by a union bound over the
$\binom{n}{2} = \Theta(m^{2/5})$ pairs $a \not = b$.
\end{proof}

\begin{thm}
The logic $\BPIFPJ$ captures \BPP.
\end{thm}
\begin{proof}
$\BPIFPJ$ is contained in \BPP, because
a randomised polynomial time algorithm can interpret the random relations by using its random
bits.

For the other direction, let $\mathcal{Q}$ be a Boolean query in
\BPP. This means that there is a randomised polynomial time algorithm
$M$ that decides the query $\mathcal{Q}_\leq$ of ordered expansions of
structures in $\mathcal{Q}$. We may view the (polynomially many)
random bits used by $M$ as part of the input. Then it follows from the
Immerman-Vardi Theorem that there is a $\BPIFP$-sentence $\psi_M$
defining $\mathcal{Q}_\leq$. Note that, by the definition of
$\mathcal{Q}_\leq$, this sentence is order-invariant.  We replace
every occurrence of $\leq$ in $\psi_M$ by the formula
$\varphi_\leq(x,y)$ of Lemma~\ref{lem:ordlem}, which with high
probability defines a linear order on the universe.
\end{proof}

It is easy to see that $\BP\IFPC$ is also contained in $\BPP$ and that
$\IFP(\Rescher)\leqq\IFPC$. Thus we
get the following corollary.

\begin{cor}
  $\BP\IFPC=\BP\IFP(\Rescher)$, and both capture $\BPP$. %
\end{cor}

\begin{rem}
  Lemma~\ref{lem:ordlem} also implies that
  $\BP\Linf(\Rescher)\equiv\BP\Cinf$, because, in the presence of an
  ordering, a quantifier of the form $\exists^{\geq n}x\,\varphi$ may
  be spelled out as
  \[
  \bigvee_{\substack{S \subset \N\\ \betrag{S}=n}} \bigwedge_{i \in S}
  \exists x\,(\varphi_{i\text{-th}}(x)\wedge \varphi(x)),
  \]
  where $\varphi_{i\text{-th}}(x)$ defines $i$-th element
  in the linear order (cf. section~\ref{sec:sandq}).

  In fact, because the formulas $\varphi_i$ use only three distinct
  variables independent of $i$, \emph{any} query is definable in
  $\Linf$ on ordered structures, as well as on $\BP\Cinf$.
\end{rem}

\section{Summary and Open Problems}

Our main motivation for introducing randomised logics was to apply
tools from finite model theory to problems in computational complexity
theory, and possibly vice versa. Because most capturing results from
descriptive complexity remain valid when both the logic and the
complexity class they involve are randomised in the same way, our
definitions are indeed suitable for this purpose. In particular, the
capturing results by Barrington et al.~\cite{bis90} for
$\FO[+,\times]$ and Behle and Lange~\cite{bl06} for $\FO[\leq]$ and
$\FO[+]$ fall into this category.

This asks for a more detailed investigation of the expressive power of
randomised logics. For example, we have shown that $\BPFO[+]$ can not
be derandomised, while conditional derandomisation results for
dlogtime-uniform $\logic{BPAC}^0$ (cf.~\cite{vio04}) suggest that
$\BPFO[+,\times]$ might be derandomisable. As this question seems to
elude currect techniques, a first step might be to find \emph{some}
relation $R$ for which $\BPFO[R]$ is derandomisable. Note that
derandomisability of non-uniform $\logic{BPAC}^0$ implies the
existence of an infinite sequence $(R_i)_{i \geq 1}$ of relations for
which $\BPFO[R_1,R_2,\ldots]$ is derandomisable.

One obstruction to proving results about randomised logics is that,
for example, \EF~games become quite complicated on structures with
both a random and a non-random part. In~\cite{csl2011}, the first
author proves some non-definability results for $\BPFO$, namely that,
on vocabularies with only unary relations, $\BPFO$ can be
derandomised, and that the ordering relation $\leq$ can not be defined
in $\BPFO$ from its corresponding successor relation. A natural next
step would be to prove whether $\BPFO$ can be derandomised on word
models or not.

\section*{Acknowledgements}
We would like to thank Nicole Schweikardt and Dieter van Melkebeek for
helpful comments on an earlier version of this paper.

\bibliographystyle{plain}
\bibliography{rlogicbib}

\end{document}